\documentclass[sigplan,10pt]{acmart}

\usepackage{amsmath}
\usepackage{booktabs}   
\usepackage{subcaption} 
\usepackage{nicefrac}
\usepackage{tikz}
\usepackage{pgfplots}
\usepackage{mathtools}
\usepackage{relsize}
\usepackage{listings}
\usepackage{comment}

\usepackage[inline]{enumitem}
\usepackage{hyperref}

\usetikzlibrary{arrows}
\usetikzlibrary{automata}
\usetikzlibrary{calc}
\usetikzlibrary{patterns}

\pgfplotsset{
    /pgfplots/ybar legend/.style={
    /pgfplots/legend image code/.code={%
       \draw[##1,/tikz/.cd,yshift=-0.25em]
        (0cm,0cm) rectangle (6pt,0.7em);},
   },
}

\pgfplotsset{compat=1.3}

\DeclareMathOperator*{\argmax}{argmax\,}

\newenvironment{proofsketch}{
\paragraph{\normalfont \textit{Proof sketch.}}
}{\hfill$\square$}

\definecolor{linkblue}{rgb}{0.16, 0.34, 0.65}

\hypersetup{
  colorlinks=true,
  linkcolor=linkblue,
  urlcolor=linkblue,
  citecolor=linkblue,
  final,
}

\newcommand{\linkref}[1]{\hyperref[#1]{\ref{#1}}}

\let\origthelstnumber\thelstnumber
\makeatletter
\newcommand*\Suppressnumber{%
  \lst@AddToHook{OnNewLine}{%
    \let\thelstnumber\relax%
     \advance\c@lstnumber-\@ne\relax%
    }%
}

\newcommand*\Reactivatenumber[1]{%
  \setcounter{lstnumber}{\numexpr#1-1\relax}
  \lst@AddToHook{OnNewLine}{%
   \let\thelstnumber\origthelstnumber%
   \refstepcounter{lstnumber}
  }%
}

\makeatletter\if@ACM@journal\makeatother
\acmJournal{PACMPL}
\acmVolume{1}
\acmNumber{1}
\acmArticle{1}
\acmYear{2018}
\acmMonth{1}
\startPage{1}
\else\makeatother
\acmYear{2018}
\startPage{1}
\fi

\setcopyright{none}             

\begin{document}

\title[Analysis of Tries with Constant-Time Operations]
{Analysis of Concurrent Lock-Free Hash Tries with Constant-Time Operations}


\author{Aleksandar Prokopec}
\orcid{nnnn-nnnn-nnnn-nnnn}             
\affiliation{
  \institution{Oracle Labs}            
  \country{Switzerland}
}
\email{aleksandar.prokopec@gmail.com}          

\begin{abstract}
Ctrie is a scalable concurrent non-blocking dictionary data structure,
with good cache locality,
and non-blocking linearizable iterators \cite{prokopec12ctries}.
However, operations on most existing
concurrent hash tries run in $O(\log n)$ time.

In this technical report,
we extend the standard concurrent hash-tries \cite{Prokopec2011}
with an auxiliary data structure called a cache\footnote{
A complete description of the new data structure
is given in the corresponding PPoPP 2018 paper
\cite{prokopec18}.
}.
The cache is essentially an array that stores pointers
to a specific level of the hash trie.
We analyze the performance implications of adding a cache,
and prove that the running time of the basic operations becomes $O(1)$.
\end{abstract}




\maketitle

\begin{figure}[t]

\begin{minipage}{4cm}
\begin{lstlisting}[
  basicstyle=\ttfamily\scriptsize,language=scala,
  xleftmargin=10pt,mathescape
]
class SNode
  val hash: Int
  val key: KeyType
  val value: ValueType
  var txn: Any

type ANode = Array<Any>

val NoTxn

val FSNode

val FVNode
\end{lstlisting}
\end{minipage}
\begin{minipage}{4cm}
\begin{lstlisting}[
  basicstyle=\ttfamily\scriptsize,language=scala,
  xleftmargin=10pt,mathescape
]
class FNode
  val frozen: Any

class ENode
  val parent: ANode
  val parentpos: Int
  val narrow: ANode
  val hash: Int
  val level: Int
  var wide: ANode

class CacheTrie
  val root = new ANode(16)
\end{lstlisting}
\end{minipage}

\caption{Basic Cache-Trie Data Types}
\label{fig:basic-data-types}
\end{figure}
\begin{figure}[t]

\begin{lstlisting}[
  basicstyle=\ttfamily\scriptsize,language=scala,numbers=left,
  xleftmargin=10pt,mathescape
]
def lookup(key: KeyType, hash: Int, level: Int,
  cur: ANode): ValueType =
  val pos = (hash >>> level)`$\odot$`(cur.length - 1)`\label{code:lookup-position}`
  val old = READ(cur[pos])`\label{code:lookup-read-old}`
  if (old == null `$\vee$` old == FVNode)`\label{code:lookup-check-vacant}`
    return null
  else if (old `$\in$` ANode)`\label{code:lookup-check-anode}`
    return lookup(key, hash, level + 4, old)
  else if (old `$\in$` SNode)`\label{code:lookup-check-snode}`
    if (old.key == key) return old.value
    else return null
  else if (old `$\in$` ENode)`\label{code:lookup-check-enode}`
    val an = old.narrow
    return lookup(key, has, level + 4, an)
  else if (old `$\in$` FNode)`\label{code:lookup-check-fnode}`
    return lookup(key, hash, level + 4, old.frozen)`\label{code:lookup-end-def}`

def lookup(key: KeyType): ValueType =
  lookup(key, hash(key), 0, root)
\end{lstlisting}

\caption{Lookup Operation}
\label{fig:lookup}
\end{figure}
\begin{figure}[t]

\begin{lstlisting}[
  basicstyle=\ttfamily\scriptsize,language=scala,numbers=left,
  xleftmargin=10pt,mathescape
]
def insert(k: KeyType, v: ValueType, h: Int,
  lev: Int, cur: ANode, prev: ANode): Boolean =
  val pos = (h >>> lev)`$\odot$`(cur.length - 1)`\label{code:insert-position}`
  val old = READ(cur[pos])`\label{code:insert-read-old}`
  if (old == null)`\label{code:insert-check-vacant}`
    val sn = new SNode(h, k, v, NoTxn)
    if (CAS(cur[pos], old, sn)) return true`\label{code:insert-cas-null-sn}`
    else return insert(k, v, h, lev, cur, prev)
  else if (old `$\in$` ANode)`\label{code:insert-check-anode}`
    return insert(k, v, h, lev + 4, old, cur)
  else if (old `$\in$` SNode)`\label{code:insert-check-snode}`
    val txn = READ(old.txn)`\label{code:insert-read-txn}`
    if (txn == NoTxn)`\label{code:insert-check-notxn}`
      if (old.key == key)`\label{code:insert-check-same-key}`
        val sn = new SNode(h, k, v, NoTxn)
        if (CAS(old.txn, NoTxn, sn))`\label{code:insert-cas-notxn-sn}`
          CAS(cur[pos], old, sn)`\label{code:insert-cas-old-sn}`
          return true
        else return insert(k, v, h, lev, cur, prev)
      else if (cur.length == 4)`\label{code:insert-check-narrow}`
        val ppos = (h >>> (lev - 4))`$\odot$`(prev.length - 1)
        val en = new ENode(prev, ppos, cur, h, lev)
        if (CAS(prev[ppos], cur, en))`\label{code:insert-cas-expand}`
          completeExpansion(en)
          val wide = READ(en.wide)
          return insert(k, v, h, lev, wide, prev)
        else return insert(k, v, h, lev, cur, prev)
      else`\label{code:insert-wide}`
        val sn = new SNode(h, k, v, NoTxn)
        val an = createANode(old, sn, lev + 4)
        if (CAS(old.txn, NoTxn, an))`\label{code:insert-cas-notxn-an}`
          CAS(cur[pos], old, an)`\label{code:insert-cas-old-an}`
          return true
        else return insert(k, v, h, lev, cur, prev)
    else if (txn == FSNode) return false`\label{code:insert-check-fsnode}`
    else
      CAS(cur[pos], old, txn)`\label{code:insert-cas-finish-txn}`
      return insert(k, v, h, lev, cur, prev)
  else if (old `$\in$` ENode) completeExpansion(old)
  return false

def insert(k: KeyType, v: ValueType) =
  if (!insert(k, v, hash(k), 0, root, null))
    insert(k, v)
\end{lstlisting}

\caption{Insert Operation}
\label{fig:insert}
\end{figure}

\begin{figure}[t]

\begin{lstlisting}[
  basicstyle=\ttfamily\scriptsize,language=scala,numbers=left,
  xleftmargin=10pt,mathescape
]
def completeExpansion(en: ENode) =
  freeze(en.narrow)
  var wide = new Array<Any>(16)
  copy(en.narrow, wide, en.level)
  if (!CAS(en.wide, null, wide))`\label{code:expand-cas-announce-wide}`
    wide = READ(en.wide)
  CAS(en.parent[en.parentpos], en, wide)`\label{code:expand-cas-parent-wide}`

def freeze(cur: ANode) =
  var i = 0
  while (i < cur.length)
    val node = READ(cur[i])
    if (node == null)
      if (!CAS(cur[i], node, FVNode)) i -= 1`\label{code:freeze-cas-fvnode}`
    else if (node `$\in$` SNode)
      val txn = READ(node.txn)
      if (txn == NoTxn)
        if (!CAS(node.txn, NoTxn, FSNode)) i -= 1`\label{code:freeze-cas-fsnode}`
      else if (txn `$\neq$` FSNode)
        CAS(cur[i], node, txn)`\label{code:freeze-cas-finish-txn}`
        i -= 1
    else if (node `$\in$` ANode)
      val fn = new FNode(node)
      CAS(cur[i], node, fn)`\label{code:freeze-cas-fnode}`
      i -= 1
    else if (node `$\in$` FNode)
      freeze(node.frozen)
    else if (node `$\in$` ENode)
      completeExpansion(node)
      i -= 1
    i += 1
\end{lstlisting}

\caption{Freezing and Expansion}
\label{fig:complete-expand}
\end{figure}
\begin{figure}[t]

\begin{lstlisting}[
  basicstyle=\ttfamily\scriptsize,language=scala,numbers=left,
  xleftmargin=10pt,mathescape
]
type Cache = Array<Any>

class CacheNode
  val parent: Array<Any>
  val misses: Array<Int>

class CacheTrie
  val root = new ANode(16)
  var cacheHead: Cache = null

def createCache(level: Int, parent: Cache): Cache =
  val cache = new Array(1 + (1 << level))
  val misses = new Array(THROUGHPUT_FACTOR * `\#`CPU)
  cache[0] = new CacheNode(null, 8, misses)
  return cache
\end{lstlisting}

\caption{Cache Data Types and Helper Functions}
\label{fig:cache-data-type}
\end{figure}
\begin{figure}[t]

\edef\codeLookupIfIsSNode{\getrefnumber{code:lookup-check-snode}}
\edef\codeLookupEndDef{\getrefnumber{code:lookup-end-def}}

\begin{lstlisting}[
  basicstyle=\ttfamily\scriptsize,language=scala,numbers=left,
  xleftmargin=10pt,mathescape
]
def lookup(k: KeyType, hash: Int, lev: Int,
  cur: ANode, lastCachee: Any, cacheLevel: Int) =
  if (lev == cacheLevel)`\label{code:modified-lookup-inhabit-anode}`
    inhabit(cache, cur, hash, lev)
  val pos = position(cur, hash, lev)`\Suppressnumber`
  ...`\Reactivatenumber{\codeLookupIfIsSNode}`
  else if (old `$\in$` SNode)
    if (lev `$\not\in [$`cacheLevel`$,$` cacheLevel + 4`$]$`)`\label{code:modified-lookup-record-cache-miss}`
      recordCacheMiss()
    if (lev + 4 == cacheLevel)`\label{code:modified-lookup-inhabit-snode}`
      inhabit(cache, old, hash, lev + 4)
    if (old.key == key)`\Suppressnumber`
  ...`\Reactivatenumber{\codeLookupEndDef}`
      return lookup(key, hash, level + 4, old.frozen)

def fastLookup(k: KeyType): ValueType =
  val h = hash(k)
  var cache = READ(cacheHead)`\label{code:fast-lookup-read-cache}`
  if (cache == null)
    return lookup(k, h, 0, root, null, -1)
  val topLevel = countTrailingZeros(cache.length - 1)
  while (cache `$\neq$` null)
    val pos = 1 + (h`$\odot$`(cache.length - 2))
    val cachee = READ(cache[pos])`\label{code:fast-lookup-read-cache-pos}`
    val level = countTrailingZeros(cache.length - 1)
    if (cachee `$\in$` SNode)
      val txn = READ(old.txn)`\label{code:fast-lookup-read-txn}`
      if (txn == NoTxn)
        if (cachee.key == k) return cachee.value`\label{code:fast-lookup-return-value}`
        else return null`\label{code:fast-lookup-return-null}`
    else if (cachee `$\in$` ANode)
      val cpos = (h >>> level)`$\odot$`(cachee.length - 1)
      val old = READ(cachee[cpos])
      if (old == FVNode `$\vee$` old `$\in$` FNode) continue
      if (old `$\in$` SNode)
        if (READ(old.txn) == FSNode) continue
      return lookup(k, h, level, cachee, level)`\label{code:fast-lookup-call-slow-lookup}`
    cache = cache[0].parent
  return lookup(k, h, 0, root, null, topLevel)
\end{lstlisting}

\caption{Modified Lookup and the Fast Lookup Operation}
\label{fig:fast-lookup-operation}
\end{figure}
\begin{figure}[t]

\begin{lstlisting}[
  basicstyle=\ttfamily\scriptsize,language=scala,numbers=left,
  xleftmargin=10pt,mathescape
]
def inhabit(cache: Array[AnyRef], nv: Any,
  hash: Int, cacheeLevel: Int) =
  if (cache == null)
    if (cacheeLevel >= 12)
      cache = createCache(8, null)
      CAS(cacheHead, null, cache)
      inhabit(cache, nv, hash, cacheeLevel)
  else
    val length = cache.length
    val cacheLevel = countTrailingZeros(length - 1)
    if (cacheLevel == cacheeLevel)
      val pos = 1 + (hash`$\odot$`(cache.length - 2))
      WRITE(cache[pos], nv)`\label{code:inhabit-write-nv}`
\end{lstlisting}

\caption{Inhabiting the Cache}
\label{fig:inhabit}
\end{figure}
\begin{figure}[t]

\begin{lstlisting}[
  basicstyle=\ttfamily\scriptsize,language=scala,numbers=left,
  xleftmargin=10pt,mathescape
]
def recordCacheMiss() =
  val cache = READ(cacheHead)
  if (cache `$\neq$` null)
    val cn = cache[0]
    val counterId = THREAD_ID % cn.misses.length
    val count = READ(cn.misses[counterId])
    if (count > MAX_MISSES)
      WRITE(cn.misses[counterId], 0)
      sampleAndAdjustCache(cache)
    else WRITE(cn.misses[counterId], count + 1)

def sampleAndAdjustCache(cache: Array<Any>) =
  val histogram = sampleSNodesLevels()
  val best = findMostPopulatedLevel(histogram)
  val prev = countTrailingZeros(cache.length - 1)
  if (histogram[best] > histogram[prev] * 1.5)
    adjustCacheLevel(best)
\end{lstlisting}

\caption{Recording Cache Misses and Sampling}
\label{fig:record-and-sample}
\end{figure}
\appendix

\setcounter{equation}{0}

\section{Complexity Analysis}\label{sec:appendix:complexity}

In this section,
we show that the expected execution time
of the cache-trie operations is $O(1)$.
The proof consists of establishing the key depth distribution,
and bounding the expected key and cache depths.
From these bounds, we then conclude that the expected distance
from the cache to the key is constant.

\begin{definition}\label{def:level}
Depth $d$ of a key is the number of pointer indirections
required to reach an S-node by following pointers in the A-nodes minus $1$,
starting from the root.
Level $\ell$ in a 16-way cache-trie
is the number of hash code bits used to find the corresponding S-node,
and is defined as $\ell = d \cdot 4$.
\end{definition}

\begin{definition}
We say that a key with a specific hash code $h$
\emph{occupies} level $\ell$ or a depth $d = \ell / 4$ in a cache-trie
if it has a unique hash code $\ell$-prefix in the cache-trie,
but it does not have a unique $(\ell - 4)$-prefix.
\end{definition}

\begin{theorem}\label{thm:distribution}
Given a universal hash function,
and a cache-trie that contains $n + 1$ keys,
the probability that an arbitrary key occupies
a position at depth $d$ is:

\begin{equation}
p(d, n) = (1 - 16^{-d - 1})^n - (1 - 16^{-d})^n
\end{equation}

\end{theorem}

\begin{proof}
Consider an arbitrary key $x$ in the cache-trie.
There are $n$ other keys in the cache-trie,
so the key $x$ will occupy the level $\ell$
if some non-empty subset of $k$ other keys
has the same $\ell$-prefix of the hash code,
but not the same $(\ell + 4)$-prefix,
\emph{and} the other $n - k$ keys have a different $\ell$-prefix.

Next, note that, for any given set of keys $S$,
the cache-trie has the same structure regardless of the order of insertion.
Thus, we can behave as if the key $x$ was the first key inserted into the cache-trie.
The rest of the $n$ keys are then inserted as $n$ independent trials --
each trial is an independent choice of a hash code,
and can either cause a collision
(have the same $\ell$-prefix and a different $(\ell + 4)$-prefix compared to $x$),
or not cause a collision (have a different $\ell$-prefix compared to $x$).
Two keys colliding at level $\ell$ have
$\ell / 4$ identical consecutive hash code substrings of length $4$,
followed by a different substring of length $4$,
so the corresponding probability is $16^{-\ell/4} \cdot 15/16$.
The probability of a non-collision is $1 - 16^{-\ell/4}$.
We are interested in those events in which there was at least one collision,
so we count all combinations of $k$ colliding keys, where $k \geq 1$.

\begin{equation}
p(\ell, n) =
\sum_{k=1}^n \binom{n}{k}
\Big( \frac{1}{16^{\ell/4}} \cdot \frac{15}{16} \Big)^k
\Big( 1 - \frac{1}{16^{\ell/4}} \Big)^{n - k}
\end{equation}

From the identity:

\begin{equation}
\sum_{k=1}^n \binom{n}{k} P^k (1-P)^{n-k} Q^k =
(1 - P + PQ)^n - (1 - P)^n
\end{equation}

we get:

\begin{equation}
p(\ell, n) = (1 - 16^{-\ell/4-1})^n - (1 - 16^{-\ell/4})^n
\end{equation}

By Definition \ref{def:level}, $d = \ell / 4$, and the claim follows.
\end{proof}

The probability function $p(d, n)$ is a probability distribution over the depth $d$,
as shown by the following corollary.

\begin{corollary}
Let a cache-trie contain a fixed number of keys $n$.
Then $p(d, n)$ is a discrete probability distribution over the depths $d$.
\end{corollary}

\begin{proof}
From the definition of a discrete probability distribution --
the sum of probabilities across all depths is $1$:

\begin{align*}
\sum_{d=0}^\infty p(d, n) =&
(1 - \frac{1}{16})^n - 0 +
(1 - \frac{1}{16^2})^n - (1 - \frac{1}{16})^n +
\ldots \\
=& \lim_{d \to \infty}
\Big[
(1 - \frac{1}{16})^n +
(1 - \frac{1}{16^2})^n - (1 - \frac{1}{16})^n + \\
&+ \ldots +
(1 - \frac{1}{16^{d+1}})^n
\Big]
= \lim_{d \to \infty} (1 - 16^{-d-1})^n = 1
\end{align*}

Furthermore, $p(d, n)$ is positive
for a non-negative $d$.
\end{proof}

The following plot illustrates the probability $p(d, n)$
for different cache-trie sizes $n$, shown on the horizontal axis.
Four different probability curves are shown
for depths $d = 0$, $d = 1$, $d = 2$ and $d = 3$.
The probability curve for $d = 0$ is initially close to $1$,
but then quickly drops to $0$ before reaching $100$ keys.
The probability curves for larger depths start at $0$,
reach their maximum, and then descend back to $0$.
Note that the horizontal axis is logarithmic --
the peeks are roughly exponentially distanced.
\\*
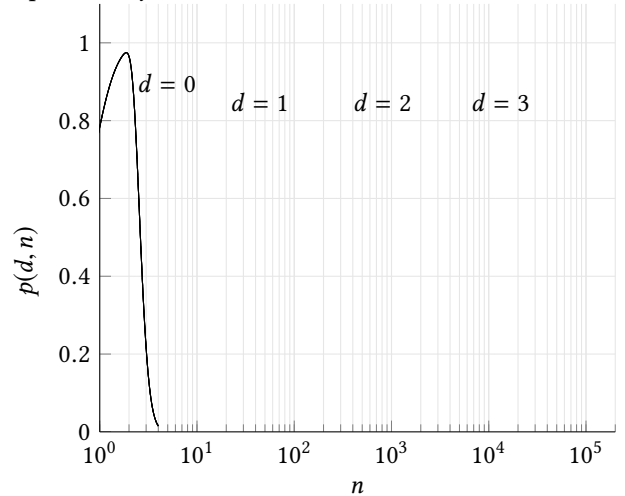
\begin{tikzpicture}
  \begin{axis}[ 
    xmode=log,
    xlabel=$n$,
    ylabel={$p(d, n)$},
    y label style={at={(axis description cs:-0.1,0.4)},rotate=0,anchor=south},
    axis x line*=bottom,
    axis y line*=left,
    ymin=0.0,
    ymax=1.1,
    xmin=1.0,
    xmax=200000,
    grid=both,
    grid style={
      lightgray!40
    },
  ]
  \addplot [dashed,no markers] gnuplot [domain=1:200000, samples=1000] {
  1.0
  };

  \addplot [smooth,no markers] gnuplot [domain=1:200000, samples=1000] {
  (1.0 - 1.0/16)^x
  };
  \node [above] at (1.6, 85) {$d = 0$};

  \addplot [smooth,no markers] gnuplot [domain=1:200000, samples=1000] {
  (1.0 - 1.0/256)^x - (1.0 - 1.0/16)^x
  };
  \node [above] at (3.8, 80) {$d = 1$};

  \addplot [smooth,no markers] gnuplot [domain=1:200000, samples=1000] {
  (1.0 - 1.0/4096)^x - (1.0 - 1.0/256)^x
  };
  \node [above] at (6.7, 80) {$d = 2$};

  \addplot [smooth,no markers] gnuplot [domain=1:200000, samples=1000] {
  (1.0 - 1.0/65536)^x - (1.0 - 1.0/4096)^x
  };
  \node [above] at (9.5, 80) {$d = 3$};
  \end{axis}
\end{tikzpicture}

The previous plot suggests that
for a given number of keys $n$ contained in the cache-trie,
most keys likely occupy a few adjacent depths.
We will construct a function $\eta(d, n)$ that captures this notion,
and then use it to construct another function $\mu(n)$ that
estimates the proportion of keys contained at the most inhabited
pair of depths.

\begin{definition}
The $\eta$ function returns the probability that a key
occupies one of the consecutive depths $d$ or $d + 1$,
and is defined as follows:

\begin{equation}
\eta(d, n) = p(d, n) + p(d + 1, n)
\end{equation}

The $\mu$ function returns the probability
that the key occupies the most inhabited pair of consecutive depths:

\begin{equation}
\mu(n) = \max_d \eta(d, n)
\end{equation}

\end{definition}

Values of $\eta(d, n)$ and $\mu(n)$ are shown below in the following plot.
It is now more apparent that, for any $n$,
a large percentage of keys occupies two consecutive depths\footnote{
Unfortunately, the arXiv server insists on rendering the Tex files on its own,
and does not accept a PDF.
At the same time, it is unaware of GNUPlot,
and does a poor job rendering some of the images.
For a better version of the images below, see the version of this document
at my homepage.
}.
\\*
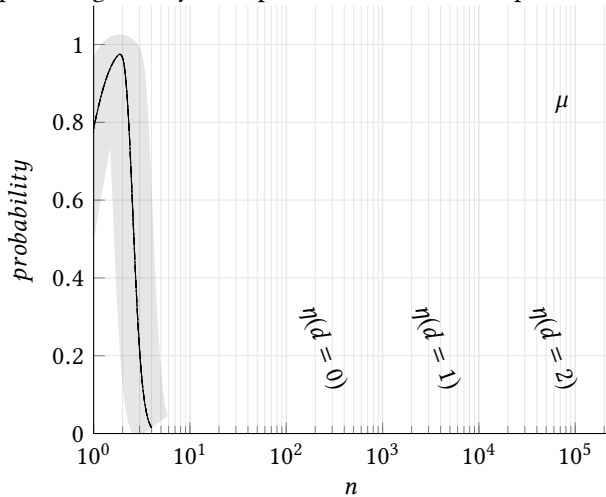
\begin{tikzpicture}
  \begin{axis}[ 
    xmode=log,
    xlabel=$n$,
    ylabel={$probability$},
    y label style={at={(axis description cs:-0.1,0.5)},rotate=0,anchor=south},
    axis x line*=bottom,
    axis y line*=left,
    ymin=0.0,
    ymax=1.1,
    xmin=1.0,
    xmax=225000,
    grid=both,
    grid style={
      lightgray!40
    },
  ]
  \addplot [dashed,no markers] gnuplot [domain=1:225000, samples=1000] {
  1.0
  };

  \addplot [style={line width=15pt,opacity=0.1},no markers] gnuplot [domain=1:225000, samples=1000] {
  0.9245
  };

  \addplot [smooth,densely dashdotted,no markers] gnuplot [domain=1:10000, samples=1000] {
  (1.0 - 1.0/256)^x
  };
  \node [above,rotate=-70] at (5.0, 20) {$\eta(d = 0)$};

  \addplot [smooth,densely dashdotted,no markers] gnuplot [domain=1:11500, samples=1000] {
  (1.0 - 1.0/4096)^x - (1.0 - 1.0/16)^x
  };
  \node [above,rotate=-70] at (7.7, 20) {$\eta(d = 1)$};

  \addplot [smooth,densely dashdotted,no markers] gnuplot [domain=1:180000, samples=1000] {
  (1.0 - 1.0/65536)^x - (1.0 - 1.0/256)^x
  };
  \node [above,rotate=-70] at (10.5, 20) {$\eta(d = 2)$};

  \addplot [smooth,no markers] gnuplot [domain=1:200000, samples=1000] {
  min(x, y) = (x < y) ? x : y
  max(x, y) = (x > y) ? x : y
  max(
    max(
      (x < 1000) ? (1.0 - 1.0/256)^x : 0,
      (x < 10000) ? (1.0 - 1.0/4096)^x - (1.0 - 1.0/16)^x : 0
    ),
    max(
      (x < 100000) ? (1.0 - 1.0/65536)^x - (1.0 - 1.0/256)^x : 0,
      max(
        (x < 150000) ? (1.0 - 1.0/1048576)^x - (1.0 - 1.0/4096)^x : 0,
        (1.0 - 1.0/16777216)^x - (1.0 - 1.0/65536)^x
      )
    )
  )
  };
  \node [above] at (11.2, 80) {$\mu$};
  \end{axis}
\end{tikzpicture}

We now substantiate the intuition
that the return values of $\mu$ are within a specific interval.

\begin{theorem}\label{thm:range}
When the number of keys $n$ tends to infinity,
values of $\mu$ are inside the interval $\langle 0.8745, 0.9746 \rangle$.
\end{theorem}

\begin{proof}
We first prove the upper bound of the interval.
The upper bound must be greater than or equal to
all the maxima of $\eta(n, d)$.

\begin{align*}
0
=& \frac{\partial \mu(n)}{\partial n} =
\frac{\partial \eta(n, d)}{\partial n} =
\frac{\partial}{\partial n}
\big[(1 - \frac{1}{16^{d+2}})^n - (1 - \frac{1}{16^d})^n \big]
\\
=& (1 - \frac{1}{16^{d+2}})^n \ln(1 - \frac{1}{16^{d+2}}) -
(1 - \frac{1}{16^d})^n \ln(1 - \frac{1}{16^d})
\end{align*}

We get the following set of maxima, parametrized by $d$:

\begin{equation}
n_{max}(d) =
\big(
\ln \frac{\ln (1 - 16^{-d})}{\ln(1 - 16^{-d - 2})}
\big)
\cdot
\big(
\ln \frac{1 - 16^{-d - 2}}{1 - 16^{-d}}
\big)^{-1}
\end{equation}

We can now compute the value of the maximum when $n$ tends to infinity.
Note that, for a specific $d$, $\eta(n, d)$ has a single maximum.
Furthermore, $n_{max}(d)$ grows monotonically with $d$.
Consequently, when $n$ tends to infinity, $d$ also tends to infinity.
By puting $n_{max}(d)$ back into the expression for $\eta(n, d)$,
and substituting $16^{-d} = t$, we get:


\begin{equation}
\mu_{upper} = \lim_{t \to 0} (1 - \frac{t}{256})^{n_{max}(t)} - (1 - t)^{n_{max}(t)}
\end{equation}

This limit can be easily simplified
(we do not show the steps for brevity),
and we get the following upper bound:

\begin{equation}
\mu_{upper} = 2^{-8/255} - 2^{-2048/255} \doteq 0.9746
\end{equation}

The upper bound is in itself not extremely useful,
since we know that $\mu \leq 1$.
The lower bound is more important,
since it mandates the minimum proportion of keys that are close to the cache.
As seen in the earlier plot for $\mu(n)$,
the minimums occur when two $\eta(n, d)$ curves meet
for two adjacent depths $d$.
The number of keys $n$ for which this happens
is the solution to the following equation:

\begin{align*}
\eta(n, d + 1) =& \eta(n, d) \\
(1 - 16^{-d-3})^n - (1 - 16^{-d-1})^n =& (1 - 16^{-d-2})^n - (1 - 16^{-d})^n
\end{align*}

This equation does not have an algebraic solution.
Fortunately, we are only interested in how $\eta$ behaves asymptotically
for large $d$.
Substituting $x = n \cdot 16^{-d}$, we get:

\begin{align*}
\lim_{d \to \infty} \eta(n, d + 1) =& \lim_{d \to \infty} \eta(n, d) \\
\Big[ \sqrt[\mathlarger{16^{3}}]{ \frac{1}{e} } \Big]^{x} -
\Big[ \sqrt[\mathlarger{16^{1}}]{ \frac{1}{e} } \Big]^{x} =&
\Big[ \sqrt[\mathlarger{16^{2}}]{ \frac{1}{e} } \Big]^{x} -
\Big[ \frac{1}{e} \Big]^{x}
\end{align*}

This equation also does not have an algebraic solution,
but we got rid of $d$, so we can solve it numerically.
The solution $x_0 \doteq 34.315$ gives us the following lower bound:

\begin{equation}
\mu_{lower} = \lim_{d \to \infty} \eta(x_0 \cdot 16^d, d) \doteq 0.874553
\end{equation}

The solution given above is approximate,
but it can be made arbitrarily precise,
and is greater than $0.8745$.
\end{proof}

Theorem \ref{thm:range} implies that there always exists
a depth $d$ such that a large percentage of the keys occupies
the depth $d$ or $d + 1$.
If the cache data structure targets this depth,
then lookups and updates for those keys take $O(1)$ time.
However, this does not yet prove the $O(1)$ bound on the expected running time --
it is possible that some small percentage of keys
are a non-costant number of levels deeper than the cache.
In what follows, we prove that this is not the case --
all the remaining keys are expected to occupy depths
that are at most a constant number of levels away from the cache,
regardless of the total number of keys $n$.

\begin{lemma}\label{lem:exponential-sum}
The following sum:

\begin{equation}
S_1(t) = \sum_{j=1}^t (1 - 16^{-j})^{16^t}
\end{equation}

converges and is less than $\frac{1}{e - 1}$ when $t$ tends to infinity.
\end{lemma}

\begin{proof}
To show this, we inspect the final term of the sum,
and conclude that it converges to $\frac{1}{e}$ for large $t$.
We then bound the sum with a geometric series.

\begin{align*}
\lim_{t \to \infty} S_1(t) =&
\lim_{t \to \infty} \sum_{j=1}^t (1 - 16^{-j})^{16^t}
\\ =&
\lim_{t \to \infty} \ldots + (1-16^{-t+1})^{16^t} + (1-16^{-t})^{16^t}
\\ =&
\lim_{t \to \infty} \ldots + (1-16^{-t+1})^{16^{t - 1} \cdot 16} + (1-16^{-t})^{16^t}
\\ =&
\ldots + \Big(\frac{1}{e}\Big)^{256} + \Big(\frac{1}{e}\Big)^{16} + \frac{1}{e}
\\ <&
\ldots + \Big(\frac{1}{e}\Big)^{4} + \Big(\frac{1}{e}\Big)^{3} + \Big(\frac{1}{e}\Big)^{2} + \frac{1}{e}
\\ =&
\frac{1}{e - 1}
\end{align*}

\end{proof}

\begin{lemma}\label{lem:root-sum}
The following sum:

\begin{equation}
S_2(t) = \smashoperator{\sum_{j=\log_{16}n}^\infty} \big[ 1 - (1 - 16^{-j})^n \big]
\end{equation}

converges and is less than $\frac{1}{e - 1}$ when $n$ tends to infinity.
\end{lemma}

\begin{proof}
We rely on the limit when $n$ tends to infinity to simplify the sum:

\begin{align*}
\smashoperator{\lim_{n \to \infty}} S_2(t) =&
\lim_{n \to \infty}
\smashoperator{\sum_{k=1}^\infty} \big[ 1 - (1 - 16^{-k - \log_{16}n})^n \big]
\\ =&
\lim_{n \to \infty}
\smashoperator{\sum_{k=1}^\infty}
\big[ 1 - (1 - 16^{-k} \cdot 16^{-\log_{16}n})^{n \cdot 16^k \cdot 16^{-k}} \big]
\\ =&
\smashoperator{\sum_{k=1}^\infty} 1 - \Big( \frac{1}{e} \Big)^{16^{-k}}
\end{align*}

We note that $e^{-k} \geq 1 - e^{-16^{-k}}$ when $k > 0$:

\begin{align*}
e^{-k} \geq& 1 - e^{-16^{-k}} \\
e^{\log_{16}x} \geq& 1 - e^{-x} \quad \quad \quad \quad k = -\log_{16}x \\
x^{(\ln 16)^{-1}} \geq& 1 - e^{-x} \\
x \geq& (1 - e^{-x})^{\ln 16}
\end{align*}

The upper bound for $S_2$ follows:

\begin{equation}
\smashoperator{\lim_{n \to \infty}} S_2(t) \leq
\sum_{k=1}^\infty e^{-k} = \frac{1}{e - 1}
\end{equation}

\end{proof}

\begin{theorem}\label{thm:expected-time}
In a cache-trie that contains $n$ keys,
the expected depth of a key is
bound by $\Theta(\log n)$.
Moreover, the expected depth is exactly $E[d](n) = \log_{16} n + O(1)$.
\end{theorem}

\begin{proof}
From the definition of the expected value
of a random variable,
for a specific choice of $n$:

\begin{align*}
E[d](n) =&
\sum_{d=0}^\infty d \cdot p(d, n)
\\ =&
\sum_{d=0}^\infty d \cdot \big[ (1 - 16^{-d-1})^n - (1 - 16^{-d})^n \big]
\\ =&
\,\, 0 \cdot p(0, n) +
(1 - 16^{-2})^n - (1 - 16^{-1})^n
\\&+
2 \cdot (1 - 16^{-3})^n - 2 \cdot (1 - 16^{-2})^n
\\&+
3 \cdot (1 - 16^{-4})^n - 3 \cdot (1 - 16^{-3})^n + \ldots
\\ =&
\lim_{d \to \infty} 
\big[
d \cdot (1 - 16^{-d-1})^n + \sum_{j=1}^d (1 - 16^{-j})^n
\big]
\\ =&
\lim_{d \to \infty}
\sum_{j=1}^d (1 - 16^{-d-1})^n - (1 - 16^{-j})^n
\\ =&
\sum_{j=1}^\infty 1 - (1 - 16^{-j})^n
\end{align*}

We separate the last sum into two parts
at $j = \lceil \log_{16}n \rceil$:

\begin{equation}\label{eq:expectation-separated}
E[d](n) =
\smashoperator{\sum_{j=1}^{\lceil \log_{16}n \rceil}}
\big[ 1 - (1 - 16^{-j})^n \big] +
\smashoperator{\sum_{\quad \quad j=\lceil \log_{16}n \rceil +1}^\infty}
\big[ 1 - (1 - 16^{-j})^n \big]
\quad
\end{equation}

Consider the first sum,
which consists of $\lceil \log_{16}n \rceil$ terms.
We use Lemma \ref{lem:exponential-sum}
to compute the following lower bound for the expected depth,
noting that the second sum in (\ref{eq:expectation-separated}) is always positive
and cannot affect the lower bound:



\begin{equation}\label{eq:depth-lower-bound}
E[d](n) \geq
\smashoperator{\sum_{j=1}^{\lceil \log_{16}n \rceil}} 1 - (1 - 16^{-j})^n
\geq
\log_{16}n - \frac{1}{e - 1}
\end{equation}

Coming back to (\ref{eq:expectation-separated}),
the first sum is upper bound by $\lceil \log_{16} n \rceil$.
The second sum is less than $(e - 1)^{-1}$ by Lemma \ref{lem:root-sum},
resulting in the following upper bound:

\begin{equation}\label{eq:depth-upper-bound}
E[d](n) \leq \lceil \log_{16}n \rceil + \frac{1}{e - 1}
\end{equation}

From (\ref{eq:depth-lower-bound}) and (\ref{eq:depth-upper-bound}),
$E[d](n) = \log_{16}n + O(1)$ follows.
\end{proof}

\begin{theorem}\label{thm:constant-distance}
For any cache-trie,
the expected distance between the cache depth
and the key depth is $O(1)$.
\end{theorem}

\begin{proof}
To emphasize the fact that we analyze the depths for a fixed number of keys $n$,
we introduce a helper function $\eta_n(d)$, defined as follows:

\begin{equation}
\eta_n(d) = \eta(n, d)
\end{equation}

We assume that the cache depth $E[d_{cache}]$ is chosen
through an \emph{unbiased sampling process} \cite{cochran1977sampling}.
By definition, this implies
that the expected sampled value $E[d_{sample}]$
of the most inhabited depth pair $d$
corresponds to the true value of the most inhabited pair.
The sampling strategy and the sample size
influence the variance of the sample,
but not its expected value.

\begin{equation}
E[d_{cache}](n) = E[d_{sample}](n) = \argmax_{d \in \mathbb{N}_0} \eta_n(d)
\end{equation}

Therefore, the expected cache depth is equal
to the true value of the most inhabited depth.
It remains to show that
the expected key depth
$E[d](n)$ is $O(1)$ levels away from
$\argmax_d \eta_n(d)$, i.e. the most inhabited depth pair.
We will show that the expected cache depth
is at most a small number of steps away from $\log_{16} n$.
To do this, we consider how the function $\eta_n(d)$
behaves around the value 
$d = \log_{16}(n) - 2$ when $n$ is large.




\begin{align*}
\smashoperator{\lim_{n \to \infty}} \eta_n(\log_{16} (n) - 2) =&
\smashoperator{\lim_{n \to \infty}} \big[
(1 - \frac{1}{n})^n - (1 - \frac{256}{n})^{256^{-1} \cdot n \cdot 256}
\big]
\\ =&
\frac{1}{e} - \Big( \frac{1}{e} \Big)^{256} \doteq \frac{1}{e}
\end{align*}

Now, consider the expression for the first derivative of $\eta_n$:

\begin{align*}
\frac{\partial \eta_n(d)}{\partial d} =& \,
n \cdot (1 - 16^{-d-2})^{n-1} \cdot 16^{-d-2} \cdot \ln 16
\\ &-
n \cdot (1 - 16^{-d})^{n-1} \cdot 16^{-d} \cdot \ln 16
\end{align*}

The first derivative of $\eta_n(d)$ is positive at $d = \log_{16} (n) - 2$:

\begin{equation*}
\lim_{n \to \infty}
\left. \frac{\partial \eta_n(d)}{\partial d} \right\vert_{\log_{16} (n) - 2} =
\big[
\frac{1}{e} - 256 \cdot \Big( \frac{1}{e} \Big)^{256}
\big] \cdot \ln 16 > 0
\end{equation*}

Similarly, we inspect $\eta_n(d)$ at $d = \log_{16} n$:

\begin{align*}
\smashoperator{\lim_{n \to \infty}} \eta_n(\log_{16} n) =&
\smashoperator{\lim_{n \to \infty}} \big[
(1 - \frac{1}{256n})^{256n \cdot 256^{-1}} - (1 - \frac{1}{n})^n
\big]
\\ =&
\Big( \frac{1}{e} \Big)^{256^{-1}} - \frac{1}{e} \doteq 1 - \frac{1}{e}
\end{align*}

The first derivative is negative at $d = \log_{16} n$:

\begin{equation*}
\lim_{n \to \infty}
\left. \frac{\partial \eta_n(d)}{\partial d} \right\vert_{\log_{16} n} =
\big[
\frac{1}{256} \cdot \Big( \frac{1}{e} \Big)^{256^{-1}}  - \frac{1}{e}
\big] \cdot \ln 16 < 0
\end{equation*}

We conclude that the value $d$
for which $\eta_n(d)$ achieves its maximum value
must be within $\langle \log_{16}(n) - 2, \,\log_{16} n \rangle$,
as illustrated by the following plot
(for the purposes of illustration, we picked $n = 1000$,
but the plot is similar for any choice of $n$).
\\*
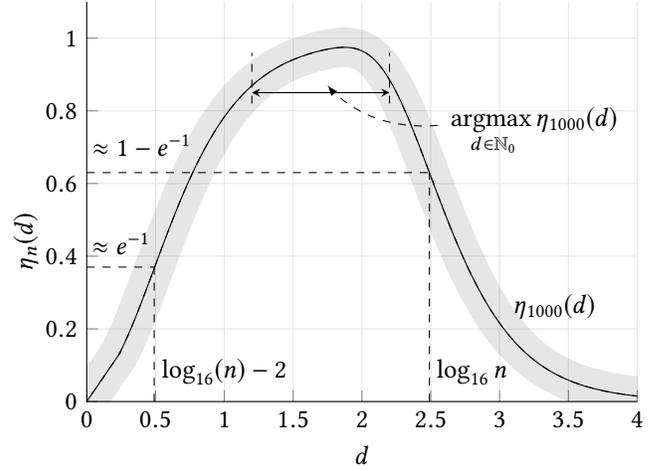
\begin{tikzpicture}
  \begin{axis}[ 
    xlabel=$d$,
    ylabel={$\eta_n(d)$},
    y label style={at={(axis description cs:-0.07,0.4)},rotate=0,anchor=south},
    axis x line*=bottom,
    axis y line*=left,
    ymin=0.0,
    ymax=1.1,
    xmin=0.0,
    xmax=4.0,
    grid=both,
    grid style={
      lightgray!40
    },
    height=6.9cm,
    width=8.9cm,
    xtick={0, 0.5, 1.0, 1.5, 2.0, 2.5, 3.0, 3.5, 4.0},
  ]
  \addplot [dashed,no markers] gnuplot [domain=0:4, samples=1000] {
  1.0
  };

  \addplot [dashed,mark=none] coordinates {(2.49, 0) (2.49, 0.63)};
  \addplot [dashed,mark=none] coordinates {(0, 0.63) (2.49, 0.63)};
  \addplot [dashed,mark=none] coordinates {(0.49, 0) (0.49, 0.37)};
  \addplot [dashed,mark=none] coordinates {(0, 0.37) (0.49, 0.37)};
  \node [above] at (280, 2) {$\log_{16} n$};
  \node [above] at (100, 2) {$\log_{16}(n) - 2$};
  \node [above] at (25, 38) {$\approx e^{-1}$};
  \node [above] at (40, 65) {$\approx 1 - e^{-1}$};
  \addplot [stealth-stealth,mark=none] coordinates {(1.20, 0.85) (2.20, 0.85)};
  \addplot [dashed,mark=none] coordinates {(2.20, 0.82) (2.20, 0.96)};
  \addplot [dashed,mark=none] coordinates {(1.20, 0.82) (1.20, 0.96)};
  \node [above] at (325, 65) {$\underset{d \in \mathbb{N}_0}{\argmax} \eta_{1000}(d)$};
  \addplot [-latex,dashed,smooth,tension=1,mark=none] coordinates {(2.55, 0.76) (2.10, 0.78) (1.75, 0.87)};

  \addplot [style={line width=15pt,opacity=0.1},no markers] gnuplot [domain=0:4, samples=1000] {
  0.9245
  };

  \addplot [smooth,no markers] gnuplot [domain=0:4, samples=1000] {
  (1.0 - 1.0/16^(x+2))^1000 - (1.0 - 1.0/16^(x))^1000
  };
  \node [above] at (340, 20) {$\eta_{1000}(d)$};
  \end{axis}
\end{tikzpicture}

By Theorem \ref{thm:range},
$\mu(n) = \max_d \eta_n(d)$ must be within
the interval $\langle 0.8745, 0.9746 \rangle$,
which is above the values $\frac{1}{e}$ and $1 - \frac{1}{e}$.
The most inhabited depth pair
must therefore also be within the same interval
$\langle \log_{16}(n) - 2, \,\log_{16} n \rangle$.
Consequently:

\begin{equation}
\argmax_{d \in \mathbb{N}_0} \eta_n(d) = \log_{16}(n) + O(1)
\end{equation}

By Theorem \ref{thm:expected-time},
the key depth $E[d](n)$ is $\log_{16}(n) + O(1)$.
Therefore, $E[d_{cache}](n) = E[d](n) + O(1)$.
\end{proof}

\begin{corollary}
Expected execution time of cache-trie lookup, insert and remove operations is $O(1)$.
\end{corollary}

\begin{proof}
Direct consequence of the Theorem \ref{thm:constant-distance},
and the fact that the cache-trie operations use the cache to find keys.
\end{proof}

\begin{corollary}[Memory Footprint]
The expected memory footprint of a
cache-augmented cache-trie is $O(n)$.
\end{corollary}

\begin{proof}
Theorem \ref{thm:expected-time} implies that
the expected path from the root the key is $O(\log n)$.
Therefore, the expected memory footprint must be less
than some complete $16$-way tree of depth $O(\log n)$, which is $O(n)$.
Furthermore, by Theorem \ref{thm:constant-distance},
the cache-level is expected to be a constant number of levels away,
so its expected memory footprint is $O(n)$ by the same argument.
\end{proof}

\section{Correctness Proofs}

We start by defining some preliminary notions,
and then showing that cache-trie operations are safe.
While proving safety, we develop sufficient foundation
to easily show linearizability as well.
After that, we show lock-freedom.
The proofs assume the absence of the cache extension.
We then prove that, by extending the cache-trie with the cache data structure,
none of the previously proven properties change.

We start with some basic definitions.

\begin{definition}[Data Types]
A \emph{single node} (\verb$SNode$)
is a node that holds a single key and a transactional marker \verb$txn$.
For a node $sn$ that holds the key $k$,
the relation $key(sn, k)$ holds.
An \emph{array node} (\verb$ANode$)
is a node that contains a sequence of pointers to other nodes
or \verb$null$ entries.
A \emph{narrow array node} contains $4$ pointers or \verb$null$s.
A \emph{wide array node} contains $16$ pointers or \verb$null$s.
For an array node $an$,
$length(an)$ is the number of pointers or \verb$null$s that it contains,
and $array(an, i)$ is the the entry at the index $i$.
A \emph{frozen node} (\verb$FNode$) is a node that wraps an array node.
For a frozen node $fn$, $unwrap(fn)$ is the node that it wraps.
A \emph{frozen single node} (\verb$FSNode$) is a marker
that denotes that a single node should not be modified.
A \emph{frozen vacant node} (\verb$FVNode$) is a marker
that denotes that an array node pointer is empty,
and should no longer be modified.
An \emph{expansion node} (\verb$ENode$) is a node that 
denotes that a narrow array node must be replaced with a wide one,
and it holds a narrow array node, and a corresponding wide array node.
For an expansion node $en$, $unwrap(en)$ is the narrow array node
that it points to.
Every node $n$ can be assigned to a level $\ell$,
and this is denoted as $n_\ell$.
Every node $n$ can be additionally assigned to a sequence of bits $p$,
and this is denoted as $n_{\ell,p}$.
\end{definition}

\begin{definition}[Child Node]
For an array node $an_\ell$ at level $\ell$ and an index $i$,
the pointer $child(an, i)$ is defined as follows:

\begin{equation}
child(an, i) = \begin{cases}
null & array(an, i) \in \{ \texttt{null} \} \\
cn & cn = array(an, i) \in \texttt{ANode} \\
unwrap(fn) & fn = array(an, i) \in \texttt{FNode} \\
unwrap(en) & en = array(an, i) \in \texttt{ENode} \\
null & array(an, i) \in \texttt{FVNode}
\end{cases}
\end{equation}

For convenience, we overload the $child$ relation for keys $k$.
For an array node $an_\ell$ at level $\ell$ and a key $k$,
the pointer $forKey(an_\ell, k)$ is defined as:

\begin{equation}
forKey(an_\ell, k) = array(an_\ell, (h \texttt{>>} \ell) \bmod length(an_\ell))
\end{equation}

where $h = hash(k)$.

For an array node $an$ and a key $k$, 
the pointer $child(an, k)$ is defined as:

\begin{equation}
child(an, k) = \begin{cases}
null & forKey(an, k) \in \{ \texttt{null} \} \\
cn & cn = forKey(an, k) \in \texttt{ANode} \\
unwrap(fn) & fn = forKey(an, k) \in \texttt{FNode} \\
unwrap(en) & en = forKey(an, k) \in \texttt{ENode} \\
null & forKey(an, k) \in \texttt{FVNode}
\end{cases}
\end{equation}

\end{definition}

\begin{definition}[Cache-Trie]
A \emph{cache-trie} is a pointer \verb$root$ to a wide array node.
A \emph{cache-trie} state $\mathbb{S}$ is the configuration of nodes
reachable from the root by following the pointers of the nodes.
A key $k$ is contained in the state $\mathbb{S}$
if and only if a single node $sn$ with the key $k$
is reachable in the corresponding configuration.
We define the relation $hasKey(an, k)$
for a node $n$ and the key $k$ as follows:

\begin{equation}
hasKey(an, k) \Leftrightarrow
\begin{cases}
sn = child(an, k) \in \texttt{SNode} \wedge key(sn) = k \\
cn = child(an, k) \in \texttt{ANode} \wedge hasKey(cn, k)
\end{cases}
\end{equation}

\end{definition}

\begin{definition}[Validity]
Let $\epsilon$ denote an empty sequence of bits,
and $an_{x,y}$ denote an \verb$ANode$.
A cache-trie that respects the following invariants is called \emph{valid}:
\\*
\textbf{INV1}
$\texttt{root} = an_{0,\epsilon} \in \texttt{ANode} \wedge
length(\texttt{root}) = 16$
\\*
\textbf{INV2}
$child(an_{\ell,p}, i) \in
\{ an_{\ell+4,p \cdot i}, \texttt{null} \} \cup \texttt{SNode}$
\\*
\textbf{INV3}
$child(an_{\ell,p}, i) = sn \in \texttt{SNode}
\Leftrightarrow
hash(key(sn)) = p \cdot i \cdot s$
\end{definition}

\begin{definition}[Abstract Set]
An \emph{abstract set} $\mathbb{A}$ is a mapping
$\mathbb{A} : K \rightarrow \{ \top, \bot \}$,
where $K$ is the set of all keys,
and which is true for the keys that are present in the abstract set.
\emph{Abstract set operations} are:

\begin{itemize}
\item
$lookup(\mathbb{A}, k) = \top \Leftrightarrow k \in \mathbb{A}$
\item
$insert(\mathbb{A}, k) =
\mathbb{A}' : k \in \mathbb{A}' \wedge
\forall k', k' \in \mathbb{A} \Rightarrow k' \in \mathbb{A}'$
\item
$remove(\mathbb{A}, k) =
\mathbb{A}' : k' \not\in \mathbb{A}' \wedge
\forall k' \neq k, k' \in \mathbb{A} \Rightarrow k' \in \mathbb{A}'$
\end{itemize}
\end{definition}

\begin{definition}[Consistency]
A cache-trie state $\mathbb{S}$
is \emph{consistent} with an abstract set $\mathbb{A}$
if and only if $\forall k, k \in \mathbb{A} \Leftrightarrow hasKey(\texttt{root}, k)$.
The cache-trie lookup on the state $\mathbb{S}$
is consistent with an abstract set lookup
if and only if for all keys $k$
it returns the value $lookup(\mathbb{A}, k)$,
where $\mathbb{A}$ is consistent with $\mathbb{S}$.
The cache-trie insert and remove on the state $\mathbb{S}$
are consistent with an abstract set insert and remove, respectively,
if an only if for all keys $k$
they change the cache-trie to a new state $\mathbb{S}'$,
such that $\mathbb{S}$ is consistent with an abstract set $\mathbb{A}$,
and $\mathbb{S}'$ is consistent with an abstract set $\mathbb{A}'$,
and $insert(\mathbb{A}, k) = \mathbb{A}'$, and $remove(\mathbb{A}, k) = \mathbb{A}'$,
respectively.
\end{definition}

Now that we have the basic definitions in place, we can state and prove
the safety property.

\begin{theorem}[Safety]\label{thm:safety}
At all times $t$,
a cache-trie is valid and consistent with some abstract set $\mathbb{A}$.
Cache-trie operations are always consistent with abstract set operations.
\end{theorem}

Before proving this theorem, we state and prove several lemmas.

\begin{definition}[Frozen Nodes]
A single node $sn$ is \emph{frozen} if its \verb$txn$ field is set to \verb$FSNode$.
The \verb$FVNode$ is always frozen.
An array node $an$ is \emph{frozen} if all entries point to frozen nodes. 
\end{definition}

\begin{lemma}[Single Transaction Change]\label{lem:single-txn-change}
A single node's \verb$txn$ field is initially set to \verb$NoTxn$,
and changes its value at most once.
\end{lemma}

\begin{proof}
By inspecting the source code, we see that every \verb$SNode$
is created with \verb$txn$ set to \verb$NoTxn$,
and every \verb$CAS$ on \verb$txn$ has \verb$NoTxn$ as the expected value.
The claim follows.
\end{proof}

\begin{lemma}[End of Life]\label{lem:end-of-life}
If an array node $an$ is frozen at some time $t_0$,
then $\forall t > t_0$ none of the entries of $an$ change their value.
If a single node $sn$ gets removed from its parent at some time $t_0$,
then its \verb$txn$ field was not set to \verb$NoTxn$ at time $t_0$.
\end{lemma}

\begin{proof}
From the definition of a frozen array node,
at $t_0$ all of its entries must be either a \verb$FVNode$, \verb$FNode$
or an \verb$SNode$ with the \verb$txn$ field set to \verb$FSNode$.
Next, note that assignments to array node entries occur in \verb$CAS$ instructions.
By inspecting these \verb$CAS$ instructions in the pseudocode,
we can see that the expected value is never \verb$FVNode$ or \verb$FNode$.
Finally, note that when the expected value is an \verb$SNode$ $sn$
(lines \ref{code:insert-cas-old-sn},
\ref{code:insert-cas-old-an}
and \ref{code:insert-cas-finish-txn} in Figure \ref{fig:insert}),
the new value for the respective \verb$CAS$ is equal to $sn$'s \verb$txn$ field.
Assume that such a \verb$CAS$ succeeds.
By Lemma \ref{lem:single-txn-change},
\verb$txn$ must be \verb$FSNode$, since it can change at most once.
That is a contradiction, since none of those \verb$CAS$ instructions
has \verb$FSNode$ as the expected value.
The argument applies inductively
if $an$ points to nested \verb$ANode$s.

For the claim about the single node $sn$,
note that only the \verb$CAS$ instructions
in lines
\ref{code:insert-cas-old-sn},
\ref{code:insert-cas-old-an},
and \ref{code:insert-cas-finish-txn} of Figure \ref{fig:insert},
and the \verb$CAS$ in line \ref{code:freeze-cas-finish-txn}
of Figure \ref{fig:complete-expand},
remove a single node from its parent.
All of those instructions first check, that \verb$txn$ is not set to \verb$NoTxn$.
From Lemma \ref{lem:single-txn-change},
\verb$txn$ could not have changed back to \verb$NoTxn$ after the check,
so the claim follows.
\end{proof}

\begin{lemma}[Freezing]\label{lem:freezing}
Let a call to the \verb$freeze$ subroutine return at some time $t_0$.
Then, the array node passed to \verb$freeze$ is frozen at time $t_0$.
\end{lemma}

\begin{proof}
Assume that the claim holds inductively for every nested \verb$ANode$,
and observe an entry at a particular index $i$.
By analyzing the different cases, we see that the index $i$ gets changed
at the end of the loop iteration in \verb$freeze$
only if the respective entry is frozen.
Hence, by the time that the loop in \verb$freeze$ completes,
the respective array node is frozen.
\end{proof}

\begin{lemma}[Unreachable Frozen Nodes]\label{lem:unreachable-frozen-nodes}
Let some \verb$CAS$ instruction make an array node $an$ unreachable at time $t_0$.
Then, that $an$ was frozen at time $t_0$.
\end{lemma}

\begin{proof}
The \verb$CAS$ instruction that makes an array node unreachable
is in line \ref{code:expand-cas-parent-wide}
of the \verb$completeExpansion$ subroutine in Figure \ref{fig:complete-expand}.
From Lemma \ref{lem:freezing}, we know that the expected value
of the \verb$CAS$ instruction in line \ref{code:expand-cas-parent-wide}
is a frozen node.
\end{proof}

\begin{lemma}[Reachable Nodes]\label{lem:reachable-nodes}
If at some time $t_0$
a thread reads a node $child$ from an array node $an$
in line \ref{code:lookup-read-old} of Figure \ref{fig:lookup}
or in line \ref{code:insert-read-old} of Figure \ref{fig:insert},
then $child$ is reachable at $t_0$.
\end{lemma}

\begin{proof}
Assume the opposite -- that the node $child$ is not reachable at $t_0$.
Then, by Lemma \ref{lem:unreachable-frozen-nodes},
$an$ was frozen at $t_0$.
But that is a contradiction, since at $t_0$
the thread read an array node, not an \verb$FNode$.
\end{proof}

\begin{lemma}[Presence]\label{lem:presence}
If a thread reads a single node $sn$ at some time $t_0$ from an array node $an$,
in line \ref{code:lookup-read-old} of Figure \ref{fig:lookup}
or in line \ref{code:insert-read-old} of Figure \ref{fig:insert}.
then the relation $hasKey(\texttt{root}, key(sn))$
holds at the time $t_0$.
\end{lemma}

\begin{proof}
From Lemma \ref{lem:reachable-nodes},
we know that the single node $sn$ was reachable at time $t_0$.
The claim follows from the definition of the $hasKey$ relation.
\end{proof}

\begin{definition}[Path]
A \emph{path} $\pi(h)$ for some hash-code $h$
is a sequence of nodes from the root to a leaf,
such that:

\begin{itemize}
\item
The first node in the path is $an_{0,\epsilon} = \texttt{root}$.
\item
$\forall an_{\ell,p} \in \pi(h)$,
if $h = p \cdot i \cdot s$
and $child(an, i) = cn$,
then the next element of the path is $cn$.
\item
$\forall an_{\ell,p} \in \pi(h)$,
if $h = p \cdot i \cdot s$
and $child(an, i) = null$,
then $an$ is the last element of the path.
\end{itemize}
\end{definition}

\begin{lemma}[Path Form]\label{lem:path-form}
Let the cache-trie be in a valid state $\mathbb{S}$.
The path $\pi(h)$ for some hash-code $h = i_0 \cdot i_1 \cdot \ldots \cdot i_n \cdot s$
is finite and has the form
$an_{0,\epsilon} an_{4,i_0} \ldots an_{4n,i_0 \cdot \ldots \cdot i_n} x$,
where $x$ is either empty or $sn \in \texttt{SNode}$.
\end{lemma}

\begin{proof}
Directly from the definition and the invariants of the valid cache-trie state.
\end{proof}

\begin{lemma}[Absence I]\label{lem:absence-1}
Let $hash(k) = p \cdot i \cdot s$ for some key $k$.
Assume that at some time $t_0$
a thread is searching for a key $k$
and it reads \verb$null$ from an array node $an_{\ell,p}$
in line \ref{code:lookup-read-old} of Figure \ref{fig:lookup}
or in line \ref{code:insert-read-old} of Figure \ref{fig:insert}.
Then $hasKey(\texttt{root}, k)$ does not hold at time $t_0$.
\end{lemma}

\begin{proof}
By the inductive hypothesis,
cache-trie was valid and consistent at the time $t_0$.
By Lemma \ref{lem:reachable-nodes},
$an_{\ell,p}$ is reachable at time $t_0$.
Now, consider the path
$an_{0,\epsilon} \rightarrow an_{4,i_0} \rightarrow \ldots \rightarrow an_{\ell,p}$
consisting of nodes that the thread traversed while searching for $k$,
where $hash(k) = i_0 \cdot i_1 \cdot \ldots \cdot i_n \cdot s = p \cdot s$.
By Lemma \ref{lem:reachable-nodes},
none of those nodes could have been removed from the cache-trie
between the time they were read and the time $t_0$.
Therefore, by Lemma \ref{lem:path-form},
if the cache-trie state contains the key $k$,
then the path to the respective single node $sn$ should have the prefix
$an_{0,\epsilon} \rightarrow \ldots \rightarrow an_{\ell,p}$.
To show this, assume by contradiction
that there is a key $sn$ that does not have this path prefix.
This assumption violates the inductive hypothesis that the cache-trie
is valid and consistent.
\end{proof}

\begin{lemma}[Absence II]\label{lem:absence-2}
Assume that at some time $t_0$
a thread is searching for a key $k$
and it reads single node \verb$sn$ such that $k \neq key(sn)$.
Then, the relation $hasKey(\texttt{root}, k)$ does not hold at $t_0$.
\end{lemma}

\begin{proof}
Similar to Lemma \ref{lem:absence-1}.
\end{proof}

\begin{lemma}[Fastening]\label{lem:fastening}
Assume that the cache-trie is valid and consistent with some abstract set.

\begin{enumerate}
\item
Assume that the \verb$CAS$ instruction
in line \ref{code:insert-cas-old-an} (Figure \ref{fig:insert})
succeeds with the expected value \verb$old$ at some time $t_1$
after \verb$old$ was read from the entry \verb$pos$ of the array node \verb$cur$
in line \ref{code:insert-read-old} at time $t_0 < t_1$.
Then, $\forall t \in \langle t_0, t_1 \rangle$,
the relation $hasKey(\texttt{root}, \texttt{k})$ does not hold.

Otherwise, if the \verb$CAS$ instruction in line \ref{code:insert-cas-old-an}
does not succeed,
then there is a duration $\delta > 0$,
such that during $\langle t_0, t_0 + \delta \rangle$
the relation $hasKey(\texttt{root}, \texttt{k})$ does not hold,
but at $t_0 + \delta$ it holds.
At $t_0 + \delta$
either the \verb$CAS$
in line \ref{code:insert-cas-finish-txn} (Figure \ref{fig:insert})
or the \verb$CAS$
in line \ref{code:freeze-cas-finish-txn} (Figure \ref{fig:complete-expand})
succeeds.
\item
Assume that the \verb$CAS$ instruction
in line \ref{code:insert-cas-old-sn} (Figure \ref{fig:insert})
succeeds with the expected value \verb$old$ at some time $t_1$
after \verb$old$ was read from the entry \verb$pos$ of the array node \verb$cur$
in line \ref{code:insert-read-old} at time $t_0 < t_1$.
Then, $\forall t \in \langle t_0, t_1 \rangle$,
the relation $hasKey(\texttt{root}, \texttt{k})$ holds.

Otherwise, if the \verb$CAS$ instruction in line \ref{code:insert-cas-old-sn}
does not succeed,
then there is a duration $\delta > 0$,
such that during $\langle t_0, t_0 + \delta \rangle$
the relation $hasKey(\texttt{root}, \texttt{k})$ holds,
and continues to hold after $t_0 + \delta$,
but \verb$old$ is not reachable at $t_0 + \delta$.
At $t_0 + \delta$
either the \verb$CAS$
in line \ref{code:insert-cas-finish-txn} (Figure \ref{fig:insert})
or the \verb$CAS$
in line \ref{code:freeze-cas-finish-txn} (Figure \ref{fig:complete-expand})
succeeds.
\item
Assume that the \verb$CAS$ instruction
in line \ref{code:insert-cas-null-sn} (Figure \ref{fig:insert})
succeeds with the expected value \verb$null$ at some time $t_1$
after \verb$old$ was read from the entry \verb$pos$ of the array node \verb$cur$
in line \ref{code:insert-read-old} at time $t_0 < t_1$.
Then, there exists some duration $\delta > 0$
such that $\forall t \in \langle t_1 - \delta, t_1 \rangle$,
the relation $hasKey(\texttt{root}, \texttt{k})$ does not hold.
\end{enumerate}
\end{lemma}

\begin{proof}
We proceed casewise:

\begin{enumerate}
\item\label{enum:fastening-an}
By Lemma \ref{lem:reachable-nodes},
we know that \verb$cur$ was reachable at time $t_0$.
First assume that the \verb$CAS$ in line \ref{code:insert-cas-old-an} was successful.
Then, the \verb$txn$ field of \verb$old$ must have been set to the same value
that \verb$CAS$ succeeded with.
By Lemma \ref{lem:single-txn-change},
we know that this was the only value that was ever written to \verb$txn$.
The conclusion is that \verb$old$ was at the position \verb$cur[pos]$
during the interval $\langle t_0, t_1 \rangle$.
Furthermore, by Lemma \ref{lem:path-form},
\verb$old$ is the only single node in the cache-trie
that could hold the key \verb$k$.
Since we check that \verb$old.key$ is not equal \verb$k$,
the conclusion is that $hasKey(\texttt{root}, \texttt{k})$
does not hold during $\langle t_0, t_1 \rangle$.

Now assume that the \verb$CAS$ in line \ref{code:insert-cas-old-an} was not successful.
This means that some other thread helped by committing the \verb$txn$ value.
In other words,
there was at least one other \verb$CAS$ instruction that changed \verb$cur[pos]$
between the \verb$CAS$ in line \ref{code:insert-cas-notxn-an}
and the \verb$CAS$ in line \ref{code:insert-cas-old-an}.
By Lemma \ref{lem:end-of-life},
the first such \verb$CAS$ instruction must
have been in line \ref{code:insert-cas-finish-txn} (Figure \ref{fig:insert})
or in line \ref{code:freeze-cas-finish-txn} (Figure \ref{fig:complete-expand}).
\item\label{enum:fastening-sn}
The reasoning is similar as in the previous case --
the \verb$old$ node must have been present during $\langle t_0, t_1 \rangle$.
This time, since we checked that \verb$old.key$ is equal to \verb$k$,
$hasKey(\texttt{root}, \texttt{k})$ holds during $\langle t_0, t_1 \rangle$.
Again, if the \verb$CAS$ was not successful,
then it means that some other thread helped in commiting the transaction
during $\langle t_0, t_1 \rangle$.
\item
Since the \verb$CAS$ in line \ref{code:insert-cas-null-sn} succeeded,
the entry \verb$cur[pos]$ was \verb$null$ during $\langle t_1 - \delta, t_1 \rangle$.
Note that the node \verb$cur$ was reachable during $\langle t_0, t_1 \rangle$,
since it was not frozen
(Lemmas \ref{lem:end-of-life} and \ref{lem:unreachable-frozen-nodes}).
Therefore, by Lemma \ref{lem:path-form},
the single node with a key \verb$k$ can only appear at \verb$cur[pos]$,
and the claim follows.
\end{enumerate}
\end{proof}

\begin{lemma}[Consistency Changes]\label{lem:consistency-changes}
Assume that the cache-trie is valid and consistent
with some abstract set $\mathbb{A}$.
Then, \verb$CAS$ instructions from Lemma \ref{lem:fastening}
induce a change into a valid state that is consistent with the abstract set semantics.
\end{lemma}

\begin{proof}
Follows directly from Lemma \ref{lem:fastening}
and the definition of consistency.
\end{proof}

\begin{lemma}[Housekeeping]\label{lem:housekeeping}
Let the cache-trie be valid and consistent with some abstract set $\mathbb{A}$.
Assume that one of the \verb$CAS$ instructions in lines
\ref{code:insert-cas-notxn-sn},
\ref{code:insert-cas-expand},
and \ref{code:insert-cas-notxn-an}
of Figure \ref{fig:insert},
or in lines
\ref{code:expand-cas-announce-wide},
\ref{code:expand-cas-parent-wide},
\ref{code:freeze-cas-fvnode},
\ref{code:freeze-cas-fsnode},
and \ref{code:freeze-cas-fsnode}
of Figure \ref{fig:complete-expand},
succeed.
Then, the cache-trie remains valid and consistent with $\mathbb{A}$
after those instructions succeed.
\end{lemma}

\begin{proof}
From the definition of the $hasKey$ relation,
it is straightforward to see that node of these \verb$CAS$ instructions
are state-changing.
\end{proof}

\begin{corollary}[Invariant Preservation]\label{cor:invariants}
Cache-trie invariants are always preserved.
\end{corollary}

\begin{proof}
From Lemmas \ref{lem:consistency-changes} and \ref{lem:housekeeping}.
\end{proof}

\begin{proof}[Proof of Theorem \ref{thm:safety}]
From Lemmas
\ref{lem:presence},
\ref{lem:absence-1},
\ref{lem:absence-2},
\ref{lem:consistency-changes} and \ref{lem:housekeeping},
and Corollary \ref{cor:invariants}.
\end{proof}

\begin{theorem}[Linearizability]\label{thm:linearizability}
The operations \verb$lookup$ and \verb$insert$ are linearizable.
\end{theorem}

\begin{proof}
An operation is linearizable if,
from the perspective of all the threads in the system,
there is a single point in time during the execution of that operation
at which the cache-trie changes consistency.

We already identified the \verb$CAS$ instructions
that change the consistency in Lemma \ref{lem:consistency-changes}.
We also identified the \verb$CAS$ instructions
that do not change the consistency in Lemma \ref{lem:housekeeping}.
It remains to show that during an execution of any operation,
there is at most one successful state-changing \verb$CAS$ instruction
that is consistent with the abstract set semantics of that operation.
We consider all successful state-changing \verb$CAS$ instructions in \verb$insert$,
and show that they either induce a state-change of a concurrent operation,
or they are the last successful state-changing \verb$CAS$ instructions
in the current operation.

\begin{itemize}
\item
It is easy to see that
following a successful \verb$CAS$
in lines \ref{code:insert-cas-null-sn},
\ref{code:insert-cas-old-sn},
and \ref{code:insert-cas-old-an}
(Figure \ref{fig:insert}),
the respective \verb$insert$ invocation
ends in a finite number of steps
without running any other \verb$CAS$ instruction.
\item
Successful \verb$CAS$ instructions in lines
in line \ref{code:insert-cas-finish-txn} (Figure \ref{fig:insert})
and \ref{code:freeze-cas-finish-txn} (Figure \ref{fig:complete-expand})
induce a state-change in a concurrent \verb$insert$ operation,
since they must occur after the successful \verb$CAS$ 
in lines \ref{code:insert-cas-notxn-sn} or \ref{code:insert-cas-notxn-an},
and before the failed \verb$CAS$
in lines \ref{code:insert-cas-old-sn} or \ref{code:insert-cas-old-an}, respectively.
\end{itemize}

Finally, note that no code path in \verb$insert$ ends the operation
without at least one such successful \verb$CAS$ operation.
\end{proof}

\begin{theorem}[Lock-Freedom]\label{thm:lock-freedom}
The operations \verb$lookup$ and \verb$insert$ are lock-free.
\end{theorem}

\begin{lemma}\label{lem:finite-cas-segments}
In each operation, there is a finite number of steps
between two \verb$CAS$ instructions,
a \verb$CAS$ instruction and the entry point or a return point.
\end{lemma}

\begin{proof}
The only \verb$while$ loop in the pseudocode is in the \verb$freeze$ subroutine,
which has a finite upper bound, and executes a \verb$CAS$ whenever
it decreases its counter.

Assume, therefore, that there exists an infinite number of execution steps
caused by recursion in \verb$insert$.
The recursion never decreases the cache-trie level,
and retains the same level if and only if a \verb$CAS$ fails.
We know from Lemma \ref{lem:path-form} that each path in the cache-trie is finite,
so the level cannot grow indefinitely.
Thus, there cannot be an infinite number of execution steps without a \verb$CAS$.
\end{proof}

\begin{lemma}\label{lem:failure-implies-success}
If any \verb$CAS$ instruction $C_0$ fails,
then there was a concurrent successful \verb$CAS$ instruction $C_1$
on the same memory location
that executed between the time $t_r$ when the expected value for $C_0$ was read,
and the time $t_0 > t_r$ when $C_0$ failed.
\end{lemma}

\begin{proof}
This can be easily proven by examining all the \verb$CAS$ instructions
and their expected values. 
\end{proof}

\begin{corollary}\label{cor:finite-state-change-segments}
From Lemmas \ref{lem:finite-cas-segments}
and \ref{lem:failure-implies-success},
it follows that there is a finite number of execution steps
between any two state changes.
\end{corollary}

A change in the cache-trie state does not imply a consistency change.
For example, when some operation expands a narrow node into a wide node,
the same set of keys remain in the cache-trie.
We therefore need to bound the number of subsequent state changes
that do not cause a consistency change --
we will show that, eventually, a state change must also change
the abstract set that the cache-trie is consistent with.

\begin{definition}[Transaction Potential]
Given a cache-trie in some state $\mathbb{S}$,
the \emph{transaction potential} $P(\mathbb{S})$
is the number of single nodes whose \verb$txn$ field is set to \verb$NoTxn$.
\end{definition}

\begin{definition}[Expansion Potential]
Given a cache-trie in some state $\mathbb{S}$,
the \emph{expansion potential} $E(\mathbb{S})$
is the number of narrow array nodes.
\end{definition}

\begin{lemma}\label{lem:potential-decrease}
Non-consistency-changing \verb$CAS$ instructions
always either decrease the expansion potential or decrease the transaction potential.
\end{lemma}

\begin{proof}
Straightforward, by enumerating the \verb$CAS$es.
\end{proof}

\begin{lemma}\label{lem:finite-consistency-changes}
Consider the consistency-changing \verb$CAS$ instructions,
and all \verb$CAS$ instructions that cause a decrease
in $P(\mathbb{S}) + E(\mathbb{S})$.
There is a finite number of steps between two such instructions.
\end{lemma}

\begin{proof}
Assume that consistency-changing \verb$CAS$ instructions never occur.
By inspecting the remaining state-changing \verb$CAS$ instructions,
we see that each of them either expands a node (i.e. decreases $E(\mathbb{S})$)
or it announces a transaction (i.e. decreases $P(\mathbb{S})$).
Therefore, after a finite number of steps,
$P(\mathbb{S}) + E(\mathbb{S}) = 0$.
This is a contradiction, since now a consistency-changing \verb$CAS$ must occur.
\end{proof}

\begin{proof}[Proof of Theorem \ref{thm:lock-freedom}]
From Corollary \ref{cor:finite-state-change-segments}
and Lemmas \ref{lem:potential-decrease} and \ref{lem:finite-consistency-changes} --
since in any cache-trie state, there is a finite number of unfinished transactions
and narrow nodes, there is a finite number of steps between consistency changes.
By Theorem \ref{thm:linearizability},
each consistency change corresponds to a completed operation.
\end{proof}

\begin{definition}[Cache]
A \emph{cache node} is a node that holds
a pointer to a parent cache,
and an integer array for tracking cache misses.
A \emph{cache array} is an array of length $2^\ell + 1$,
that contains a cache node at the index $0$,
and the remaining entriese are either \verb$SNode$s or \verb$ANode$s.
We say that such a cache array is at level $\ell$.
A \emph{cache} is a pointer that is either \verb$null$ or points to a cache array.
\end{definition}

\begin{lemma}[Cachee Level]\label{lem:cachee-level}
A cache array at level $\ell$ contains pointers
to nodes that are either frozen or have \verb$txn$ different than \verb$NoTxn$,
or are reachable and at level $\ell$ of the cache-trie
(except at the index $0$).
\end{lemma}

\begin{proof}
Consider the only write to the cache
in line \ref{code:inhabit-write-nv} of Figure \ref{fig:inhabit}.
The write is preceded by a check that the cache level corresponds
the cachee level.
Any \verb$SNode$ or \verb$ANode$ in the cache-trie
that is made unreachable must be frozen by Lemma \ref{lem:unreachable-frozen-nodes},
and frozen nodes are never made reachable once they become unreachable.
Therefore, a reachable node remains at level $\ell$ of the cache.
\end{proof}

\begin{theorem}[Cache Safety]\label{thm:cache-safety}
Both fast insert and fast lookup are consistent with the abstract set semantics.
Moreover, fast insert and fast lookup are linearizable and lock-free.
\end{theorem}

\begin{proof}
We start by considering consistency.
From Lemma \ref{lem:cachee-level},
we know that a non-frozen node or a single node
whose \verb$txn$ is different than \verb$NoTxn$
is reachable in the cache-trie.

Assume that the fast lookup or fast insert calls the normal lookup or insert,
e.g. in line \ref{code:fast-lookup-call-slow-lookup}.
Since the respective array node is reachable,
safety and linearizability follow immediately by the same reasoning
as Theorems \ref{thm:safety} and \ref{thm:linearizability}.
Since the total number of execution steps is smaller
compared to an execution in which the search started at the root,
lock-freedom follows as well, by Theorem \ref{thm:lock-freedom}.

Assume that the fast lookup returns a value
in line \ref{code:fast-lookup-return-value},
or \verb$null$ in line \ref{code:fast-lookup-return-null}
of Figure \ref{fig:fast-lookup-operation}.
This is preceded by the read of the \verb$txn$ field at some time $t_0$,
in line \ref{code:fast-lookup-read-txn}.
Since \verb$txn$ is \verb$NoTxn$ at $t_0$, the corresponding node is reachable,
and by Lemma \ref{lem:path-form},
the relation $hasKey$ must hold.
Therefore, fast lookup is safe.
The read of the \verb$txn$ field is the linearization point,
so fast lookup is linearizable.
Since it takes a finite number of steps to reach that point,
fast lookup is also lock-free.

Finally, consider the case in which the node
read from the cache is frozen or has \verb$txn$ different than \verb$NoTxn$.
In this case, a fast lookup or a fast insert both fall back to a normal,
so all three properties hold, by Theorems
\ref{thm:safety}, \ref{thm:linearizability} and \ref{thm:lock-freedom}.
\end{proof}


%
%



\end{document}